\documentclass[a4paper,fleqn]{cas-dc}
\let\printorcid\relax
\usepackage{algorithm}
\usepackage{algorithmic}
\usepackage{cases}
\usepackage{makecell}
\usepackage{amsmath}
\usepackage{amsthm}
\usepackage{graphicx}
\usepackage{subfig}
\usepackage[numbers,sort&compress]{natbib}
\usepackage{amssymb}
\newtheorem{theorem}{Theorem}
\def\tsc#1{\csdef{#1}{\textsc{\lowercase{#1}}\xspace}}
\tsc{WGM}
\tsc{QE}
\tsc{EP}
\tsc{PMS}
\tsc{BEC}
\tsc{DE}

\begin{document}
\let\WriteBookmarks\relax
\def\floatpagepagefraction{1}
\def\textpagefraction{.001}

\shorttitle{Startup Delay Aware Short Video Ordering: Problem, Model, and A Reinforcement Learning based Algorithm}

\shortauthors{GAO et~al.}

                     
\title [mode = title]{Startup Delay Aware Short Video Ordering: Problem, Model, and A Reinforcement Learning based Algorithm}      

\tnotetext[1]{This work was supported in part by the National Nature Science Foundation of China under Grant 61872031, Grant 61572071, Grant 61531006, and Grant 61872331.
}

%

\author[1]{Zhipeng Gao}
\ead{zhpgao@bjtu.edu.cn}
\credit{Investigation, Methodology, Data analysis, Simulation, Writing – original draft}

\author[1]{Chunxi Li}

\ead{chxli1@bjtu.edu.cn}
\credit{Investigation, Methodology, Data analysis, Writing – original draft}

\author[1]{Yongxiang Zhao}
\cormark[1]

\ead{yxzhao@bjtu.edu.cn}
\credit{Investigation, Methodology, Data analysis, Writing – review \& editing}

\author[2]{Baoxian Zhang}
\ead{bxzhang@ucas.ac.cn}
\credit{Investigation, Methodology, Data analysis, Writing – review \& editing}

\affiliation[1]{
organization={School of Electronic and Information Engineering, Beijing Jiaotong University},
city={Beijing},
country={China}}

\affiliation[2]{organization={Research Center of Ubiquitous Sensor Networks, University of Chinese Academy of Sciences},
    city={Beijing},
    country={China}}


\cortext[cor1]{Corresponding author}



\begin{abstract}
Short video applications have attracted billions of users on the Internet and can satisfy diverse users' fragmented spare time with content-rich and duration-short videos. To achieve fast playback at user side, existing short video systems typically enforce burst transmission of initial segment of each video when being requested for improved quality of user experiences.
However, such a way of burst transmissions can cause unexpected large startup delays at user side. This is because users may frequently switch videos when sequentially watching a list of short videos recommended by the server side, which can cause excessive burst transmissions of initial segments of different short videos and thus quickly deplete the network transmission capacity.
In this paper, we adopt token bucket to characterize the video transmission path between video server and each user, and accordingly study how to effectively reduce the startup delay of short videos by effectively arranging the viewing order of a video list at the server side.
We formulate the optimal video ordering problem for minimizing the maximum video startup delay as a combinatorial optimization problem and prove its NP-hardness. We accordingly propose a Partially Shared Actor Critic reinforcement learning algorithm (PSAC) to learn optimized video ordering strategy. Numerical results based on a real dataset provided by a large-scale short video service provider demonstrate that the proposed PSAC algorithm can significantly reduce the video startup delay compared to baseline algorithms.
\end{abstract}



\maketitle

\section{Introduction}
In recent years, short video applications (e.g., TikTok~\cite{tiktok}, YouTube Shorts~\cite{youtube}, and Kwai~\cite{kwai}) have attracted billions of users, which can satisfy various users' fragmented spare time with content-rich and duration-short videos. According to a public report \cite{report}, the number of monthly active users on TikTok alone has been up to 1.4 billion in 2022 and more than 1 billion videos are watched by these users every day. However, research on short video transmissions is still rare. 
To transmit a short video to a user, existing short video systems typically use a two-phase transmission mechanism including burst-trans and slow-trans phases, which is in fact inherited from traditional online video-on-demand (VoD) systems~\cite{huang2014buffer}. 
The burst-trans phase is to send the initial segment of a video at a very fast rate for fast startup playback, and once the playback starts, the slow-trans phase is used to send the remaining data of the video at a relatively low rate almost equal the video encoding rate to ensure continuous playback while reducing the potential waste of downloaded video data due to random user departures \cite{allard2020measuring}. 

However, use of such two-phase video transmissions can cause unexpected high startup delays at user side due to some inherent characteristics of short video systems. Details are as follows.
First, short video users are very sensitive to video startup delay, typically expecting very small video startup delays less than one second \cite{zhang2022measurement}, while VoD users can tolerate startup delays of several seconds or even more. This is because a short video is just a few tens of seconds long while a traditional VoD video is typically tens of minutes long \cite{chen2021edge}. 
Second, short video users can quickly determine their preference for a video by only going through its first few screens and can easily switch a video by simply scrolling their screens \cite{zhang2022duasvs}. Therefore, a short video user may frequently switch videos when sequentially watching videos according to a video list provided by the server side, leading to high dense burst transmissions of initial segments of the videos in short time so as to deplete the network transmission capacity quickly. In this case, the short video users will suffer significantly large startup delays.

Shortening the startup delay of short videos can largely improve the quality of user experiences and it has been a critical issue for improving the performance of short video systems~\cite{nguyen2022network}. The startup delay of a video experienced by a user is the length of time interval from the time instant when the user scrolls his/her screen to request the video to the time instant when the user finishes downloading the initial segment of the video \cite{nguyen2022network}.
In this aspect, there have been some studies that can be helpful to alleviate the startup delays of short videos. In \cite{chen2021edge,mao2017survey,chen2019study,hu2021collaborative}, the authors proposed to cache short videos at edge servers in the vicinity of users, which could be costly due to the requirement of deploying many cache servers at network edge. In \cite{nguyen2022network,krishnamoorthi2014quality,chen2019energy}, the authors proposed some prefetching mechanisms, where a client is suggested to prefetch the next video’s data during the playback of the current video. These mechanisms require modifications to the client software and thus are often difficult to implement in practice. Moreover, severe waste of downloaded video data can occur in the presence of frequent video switching at user side, which is quite common in short video watching. Therefore, in-depth study is still needed to seek effective methods for reducing the startup delay of short videos. 

In this paper, we study how to effectively reduce the short video startup delay by taking advantage of the inherent characteristics of short video systems. 
On the one hand, from the perspective of short video transmissions, given the transmission capacity of a session for delivering videos from video server to a user, which can be characterized by using a token bucket for traffic shaping from the perspective of network operator \cite{kanuparthy2011shaperprobe}, 
the burst transmissions of initial segments of videos are bound to quickly consume the tokens in the bucket, which can largely hurt the session's available capacity for burst transmissions of subsequent videos. In contrast, once a video starts playing, which indicates that the transmission enters the slow-trans phase, the session's available transmission capacity can be gradually recovered as the tokens are replenished as time evolves. 
Obviously, a longer viewing time is more helpful to restoring the transmission capacity, as more tokens can be accumulated during this period.
On the other hand, from the perspective of users' viewing habits, a user usually watches short videos sequentially based on the video list received from the server side, while the viewing time of a video can be predicted
according to \cite{zhan2022deconfounding,Lin2023tree}. 
As a result, given a set of videos, each of which has a certain viewing time, a natural idea for reducing the video startup delays for their watching is to properly arrange their playback order by interlacing the videos with long viewing time and those with short viewing time. In this way, the excessively consumed network burst transmission capacity caused by a video with short viewing time can be replenished as much as possible by the subsequent video(s) with long viewing time. Accordingly, sufficient burst transmission capacity can be preserved to meet the subsequent burst transmission demands and thus reduce the video startup delay. 

Realization of the above idea for reduced startup delay involves the following three aspects. First, an effective transmission model is needed to quantitatively characterize the impact of burst-trans and slow-trans on the transmission capacity. Second, a mathematical model is needed to calculate the video startup delays for a given video list. Third, due to the huge search space of up to $N!$ ordering choices for a video list containing $N$ videos, an effective method is required to quickly find an optimized video list for minimizing the maximum startup delay.

In this paper, we use token bucket to help formulate the optimization problem of video ordering and use reinforcement learning based algorithm to solve the problem. First, we introduce a token bucket to characterize the transmission capacity of the delivery path between video server and a user, where the minimal network transmission capacity equals the token rate and the maximal (burst) network transmission capacity is determined by the token bucket capacity. 
Then, given a video list where each video has a certain viewing time, we use the change of number of tokens in the bucket to characterize the transmission process of different segments in a video and accordingly give closed form expression for the startup delay of each video. 
Following that, we formulate the video ordering problem for minimizing the maximum video startup delay as a combinatorial optimization problem and prove this problem is NP-hard. To address this problem, we propose a Partially Shared Actor-Critic learning algorithm (PSAC), which is designed based on the None-module Shared Actor-Critic learning algorithm (NSAC) in \cite{bello2016neural}. However, different from NSAC, our PSAC algorithm allows actor and critic share some functionally identical modules to reduce the amount of model parameters so as to improve the convergence speed and also startup delay performance of generated video list. 
In PSAC, the actor is for generating a video list for a given video set and calculating the maximal startup delay according to the token bucket based transmission model, while the critic is for estimating the maximal startup delay for the video list, and both of them are optimized by using policy gradient and stochastic gradient descent aiming to learn an optimal policy that can generate a video list tending to minimize the maximum startup delay. 
We conduct extensive simulations on a dataset provided by a large-scale short video service provider, which contains more than 6 million viewing records of 16000 users. The numerical results demonstrate the significantly high performance of our PSAC algorithm, which can reduce the average maximum startup delay by up to 56\%, 54\%, 45\%, and 42\% compared with RAND, which randomly orders videos, INTL, which interleaves videos according to their viewing time such that neighbor videos have the largest difference in viewing time, GRDY, which orders videos in a greedy manner, and NSAC, respectively.

The remainder of this paper is organized as follows. 
Section~\ref{rw} briefly reviews related work. 
Section~\ref{smps} describes the short video transmission model, motivation for video ordering, formulates the optimal video ordering problem under study, and prove its NP-hardness.
In Section~\ref{psac}, we propose the PSAC algorithm for optimized video ordering. Section~\ref{pe} evaluates the performance of PSAC by comparing it with baseline algorithms. Finally, Section~\ref{con} concludes this paper.

\section{Related Work}
\label{rw}
In this section, we briefly review related work in the following two aspects: Methods helpful for alleviating the startup delay of short videos and methods for machine learning for sequence reordering, which can inspire us to conduct our work in this paper.

\subsection{Startup Delay Reduction}

Edge caching and video prefetching are two main techniques that can be used to effectively reduce the startup delay of short videos. Next, we shall respectively introduce typical work belonging to these two aspects.

Edge caching works to deploy cache servers at the network edge close to users and allows users to fetch as much video data as possible from the cache servers rather than remote video servers, so as to reduce video transmission delay and thus reduce the video startup delay \cite{li2021joint}. 
Typical work in this direction is as follows. In \cite{chen2019study}, Chen et al.~suggested that edge servers caching top popular short videos can improve the quality of user experiences for short videos.
In \cite{hu2021collaborative}, Hu et al.~proposed a collaborative caching algorithm for multi-edge-server scenarios so as to allow the edge servers to jointly cache more top popular short videos to better serve users in the vicinity by maximizing the overall user benefit (including the transmission delay performance).
In \cite{chen2021edge}, Chen et al.~proposed a transcoding assisted short video caching framework that utilizes the storage and computing resources at edge servers to facilitate both video caching and transcoding for improved caching performance.
However, the cost of realizing the above edge caching methods can be very high because a lot of cache servers need to be deployed at the network edge and also much network resources need to be leased to provide such edge caching services. 
Client-side video prefetching can also be used to reduce the video startup delay by prefetching the data of the subsequent video(s) during the playback of a current video. 
Typical work in this direction is as follows. In \cite{chen2019energy}, Chen et al.~proposed a one-video prefetching mechanism, which prefetches only the next video while playing a current video. 
In \cite{zhang2020apl}, Zhang et al.~proposed a multi-video prefetching mechanism, which can prefetch multiple subsequent short videos in parallel while playing a current video, in order to improve the playback smoothness and reduce the data wastage. 
In \cite{zhang2022duasvs}, Zhang et al.~proposed a reinforcement learning (RL) based multi-video prefetching mechanism, where each video is associated with a threshold that limits the video's data amount to be prefetched, and the threshold is decided by a RL model trained with the video's watching history.
However, video prefetching methods have the following drawbacks. First, they are difficult to develop and deploy since their implementation and deployment involve modifications to client software. 
Second, video prefetching may cause serve data wastage since users may switch videos frequently or even quit from the watching.

\subsection{Machine Learning for Sequence Reordering}

Here, we briefly review machine learning (ML) based research for solving combinatorial optimization problems (COPs),
considering that the video ordering problem under study in this paper is in fact a typical COP problem, which needs to convert an input (video) sequence to another output sequence meeting a certain optimization objective while elements in the input and output sequences are identical.
Solving such COP problems by using ML methods is an emerging trend recently \cite{bengio2021machine}, as these problems are typically NP-hard and in general do not have polynomial-time optimal solutions \cite{mazyavkina2021reinforcement}.
In \cite{vinyals2015pointer}, Vinyals et al.~proposed a supervised learning (SL) based encoder-decoder structured pointer network (Ptr-Net), which can learn the conditional probability of an output sequence whose elements are the indices of the elements of the input sequence.
Following that, considering that the optimal output sequence for SL training is difficult to obtain, in \cite{bello2016neural}, Bello et al.~proposed an actor-critic structured RL framework for training the Ptr-Net, where the actor implements the Ptr-Net and the critic (constructed by an RNN and a fully-connected network) learns the expected value of the optimization objective for the output sequence of the Ptr-Net, and both are jointly trained by using policy gradient and stochastic gradient descent. 
Our PSAC algorithm in this paper also uses RL, where the difference is that, in PASC, the actor and critic are allowed to share some functionally identical modules to reduce the amount of model parameters so as to improve the convergence speed and performance.

\section{System Model and Problem Statement}
\label{smps}
In this section, we first introduce a token bucket to model the transmission of short videos and elaborate how the playback order of videos can affect their associated startup delays. Then, we formulate the video ordering problem for minimizing the maximum startup delay. Finally, we prove this problem is NP-hard.

\subsection{Video Transmission Model}
\label{vtm}
In this paper, we adopt token bucket (see Fig.~\ref{FIG:transmissionmodel}) to model the short video transmission from a video server to a user by limiting the token (generation) rate as well as the maximum burst size, considering that it has been widely used by network operators in traffic shaping~\cite{kanuparthy2011shaperprobe}.
Next, we shall introduce each of the components from right to left in Fig.~\ref{FIG:transmissionmodel} .

First, the server recommends multiple short videos at a time and sends them to the user according to a certain playback order. For each video, the server supports two rates (i.e., burst rate and normal rate) for injecting the video data into the data buffer of the token bucket. 
The burst rate is used to quickly send the initial segment of a video aiming to enable the user to startup the playback as quick as possible. Once the video starts playing, normal rate is used to transmit the rest segments to enable smooth playback while maximally avoiding data wastage due to random video switching at user side.
We assume that, the burst rate is much higher than the token rate, and the normal rate equals the video encoding rate, while the token rate is larger than or equal to the video encoding rate. 
That is, normal rate = video encoding rate $\leq$ token rate $\ll$ burst rate. In addition, we assume that each video is always played at its encoding rate. 

Then, the token bucket is used to shape the data traffic from video server to a user as follows. Assume that the token bucket has a certain capacity, denoted by $C$, while the data buffer at user side is infinitely large. New tokens are injected into the token bucket at a certain constant rate $\mu$. As long as the amount of tokens in the bucket is sufficient, the download rate experienced by the user should equal the current sending rate at the server side (i.e., either the burst rate or the normal rate)\footnote{Here, we assume that the available bandwidth on the path from video server to each user is not the bottleneck}; otherwise, the download rate equals the token rate. 

Finally, the user receives video data from the token bucket and experiences a certain startup delay. For each video, the user can only start its playback after fully receiving the initial segment of the video.  
Once the user switches to a next video by scrolling his/her screen, the server immediately interrupts the current video’s transmission (if any) and starts to transmit the new one with no delay. 
As an important factor affecting the quality of user experiences, a video's startup delay refers to the length of the time interval from the time when the user requests the video to the time when the user starts playing that video, and considering the limited token bucket capacity and token rate (e.g., 1MB and 2Mbps, respectively) 
\cite{kim2012isptrafficshaping}, the startup delays at user side can be greatly affected by the playback order arranged at the server side (see the next subsection for illustration).

\begin{figure}
\centering		\includegraphics[width=0.45\textwidth]{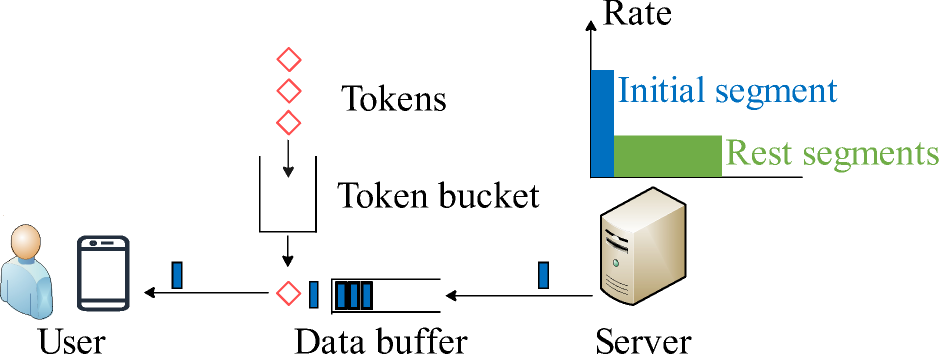}
\caption{Token bucket based transmission model.}
\label{FIG:transmissionmodel}
\end{figure}

\subsection{Observations and Motivation}
To motivate our study of optimized video ordering, here, we first present two examples to illustrate the effect of video playback order on startup delay, and then propose the key idea behind our design in this paper.

Consider video server at a time recommends eight videos (i.e., $v_1, v_2, \cdots, v_8$), each with the same encoding rate while having different predicted viewing time: The first four have short viewing time and the last four have long viewing time. Here, we assume video user watches these videos according to their predicted viewing time.
Next, let us consider two playback orders: 1) ($v_1$, $v_2$, $v_3$, $v_4$, $v_5$, $v_6$, $v_7$, $v_8$) and 2) ($v_1$, $v_5$, $v_2$, $v_6$, $v_3$, $v_7$, $v_4$, $v_8$), and show how the different playback orders affect the startup delay.
\begin{figure}
\centering		\includegraphics[width=0.45\textwidth]{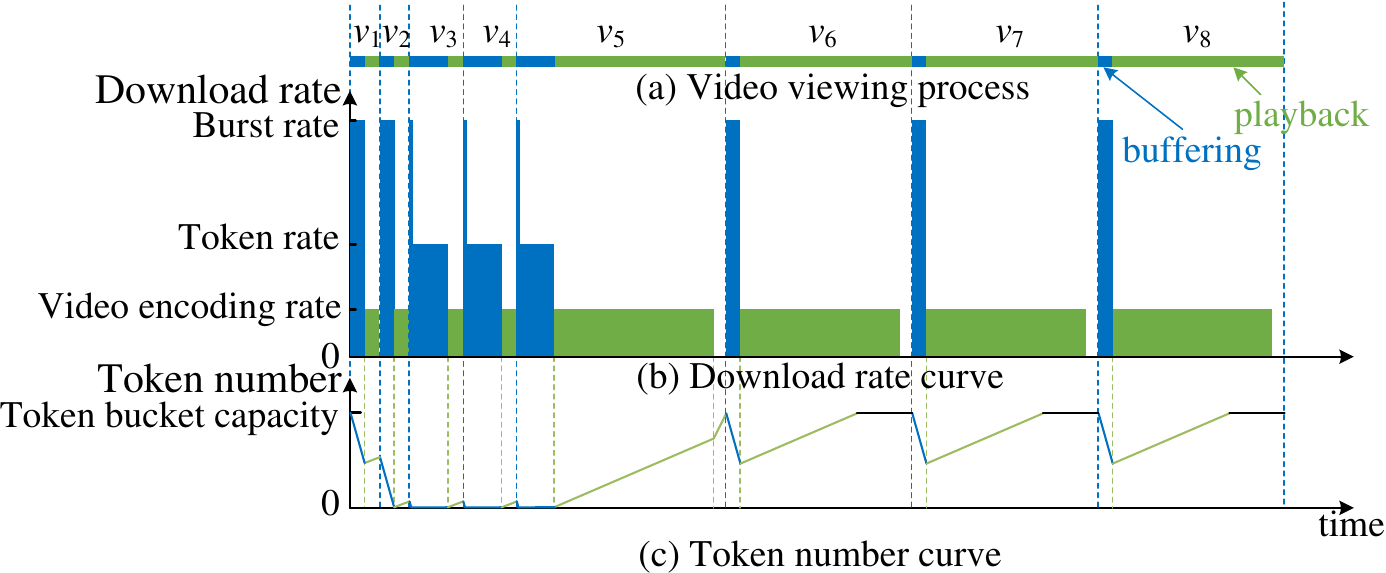}
\caption{Impact of playback order. In this example, there are eight short videos to be played and in the following order: $v_1$, $v_2$, $v_3$, $v_4$, $v_5$, $v_6$, $v_7$, and $v_8$).}
\label{FIG:example1}
\end{figure}

\begin{figure}
\centering
\includegraphics[width=0.45\textwidth]{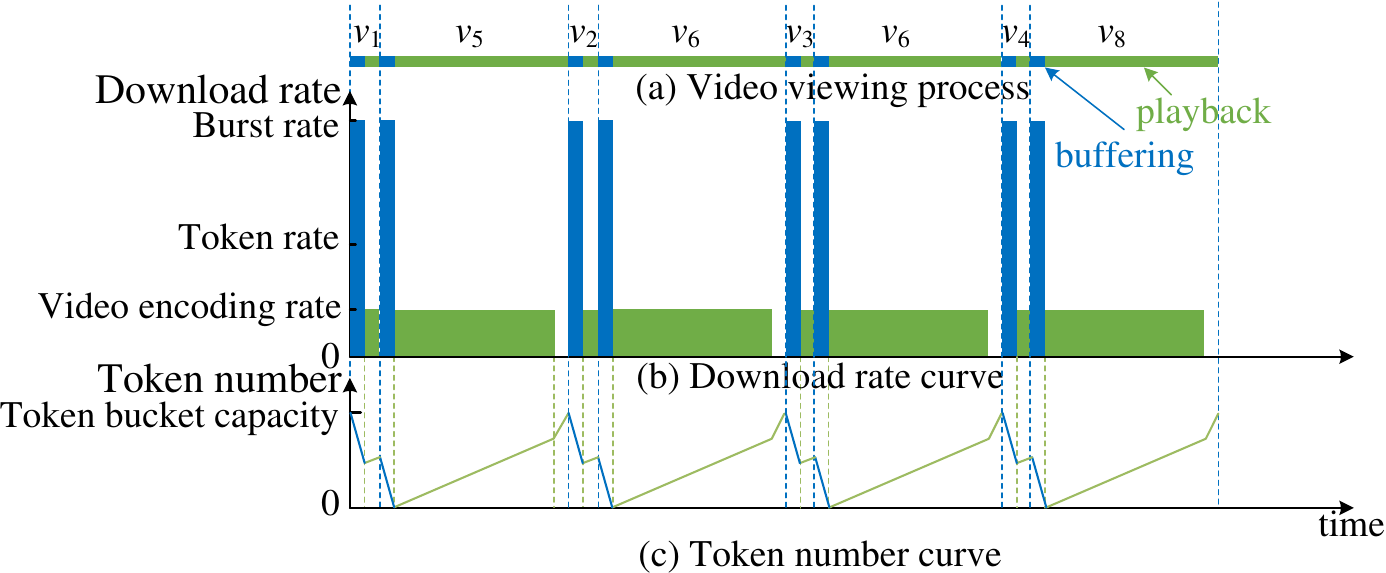}
\caption{Impact of playback order. The videos to be played are the same with those in Fig.~\ref{FIG:example1} but played in the following order: $v_1$, $v_5$, $v_2$, $v_6$, $v_3$, $v_7$, $v_4$, and $v_8$.}
\label{FIG:example2}
\end{figure}

Figs.~\ref{FIG:example1} and~\ref{FIG:example2} respectively illustrate how different playback orders can cause different cascading effects on the number of tokens in the token bucket (see subfigure (c) in both figures), user-experienced video download rate (see subfigure (b)), and video startup delay (see subfigure (a)). 
The three subfigures in either figure share the same $x$-axis, which represents the evolution of time.  
In each subfigure, each of the videos is associated with a blue block and a green block, which correspond to the deliveries of the initial segment and the rest segments of that video, respectively. Specifically, the blue blocks for a video in subfigs. (a)-(c) represent the video's startup delay, download rate of its initial segment at the user side, and the fluctuation of number of remaining tokens in the token bucket, respectively, while the green blocks in the three subfigures represent its viewing time, download rate of its rest segments, and fluctuation of number of remaining tokens in the token bucket, respectively.  
We will see that arranging those short-viewing-time videos together leads to token exhaustion, and therefore longer initial segment download time and worsening startup delays at subsequent videos.

Example I: Fig.~\ref{FIG:example1} shows the result of playback order ($v_1$, $v_2$, $v_3$, $v_4$, $v_5$, $v_6$, $v_7$, $v_8$), where those short-viewing-time videos $v_1$, $v_2$, $v_3$, $v_4$ are arranged together and thus result in rather poor startup delays. 
From this figure, it is seen that for the first two short-viewing-time videos, $v_1$ and $v_2$, the token bucket still has enough tokens to satisfy the burst transmissions of their initial segments. 
Even though, due to the short viewing time of these two videos, the number of newly generated tokens during their viewing are far less than the number of tokens consumed for the burst transmission of their initial segments.
As a result, when transmitting $v_3$, there are few remaining tokens and new tokens are generated merely at the token rate, thus leading to a quite large startup delay. Due to similar reasons, $v_4$ and $v_5$ also have large startup delays. 
After that, each subsequent video (starting from $v_5$) has a long enough viewing time to fully replenish the consumed tokens, resulting in very small startup delays for $v_6, v_7$, and $v_8$. 

Example II: Fig.~\ref{FIG:example2} shows the result of playback order ($v_1$, $v_5$, $v_2$, $v_6$, $v_3$, $v_7$, $v_4$, $v_8$), where short-viewing-time videos and long-viewing-time videos are interleaved and thus result in small startup delays for all the videos. 
From this figure, it is seen that, although the transmission of each short-viewing-time video consumes certain tokens, the remaining tokens are still sufficient to support the burst transmission of the initial segment of the next long-viewing-time video; 
In turn, each long-viewing-time video brings a long enough time to replenish enough tokens to fully support the next video’s initial segment’s burst transmission. 

Motivated by the above two examples, it should be maximally avoided to put short-viewing-time videos together, which may lead to token depletion and thus prolonged startup delay; Instead, it is good to interleave short-viewing-time and long-viewing-time videos, so that the tokens consumed for the transmission of a short-viewing-time video can be fully replenished during the viewing time of the subsequent long-viewing-time video and thus allow each video to have a small startup delay. 

Motivated by this observation, we shall formulate the video ordering problem for minimizing the maximum startup delay in the next subsection.

\subsection{Problem Formulation}

To ease the description, we list the symbols used hereafter in Table~\ref{tbl1}.

\begin{table}[width=.9\linewidth,cols=2,pos=h]
\caption{Symbols Used.}\label{tbl1}
\begin{tabular*}{\tblwidth}{@{} LLLL@{} }
\toprule
Symbols & Descriptions \\
\midrule
$V$ & Video set\\
$T_v$ & Duration of video $v$ \\
$\tau_v$ & Viewing time of video $v$ \\
$r_v$ & Encoding rate of video $v$  \\
$B_v$ & Size of initial segment of video $v$ \\
$X$ & Video list\\
$x_i$ & The video at the $i^{th}$ position of video list $X$\\
$\hat{r}$ & Burst rate of video server\\
$C$ & Token bucket capacity\\
$\mu$ & Token rate\\
$K_v$ & \makecell[cl]{Number of available tokens in the bucket \\when video $v$ starts to be transmitted}\\
$d_{v}$ & Startup delay of video $v$\\
$D$ & Maximum startup delay of the videos in $X$\\
\bottomrule
\end{tabular*}
\end{table}

First, we give some preliminary knowledge and assumptions for the transmission model in Fig.~\ref{FIG:transmissionmodel} to work.
Suppose the server recommends a set of videos at a time, denoted as $V=\{1, 2, \cdots,|V|\}$, and further, for each video $v \in V$, the server knows its duration $T_v$, encoding rate $r_v$, size of its initial segment $B_v$, and its viewing time $\tau_v$, which could be predicted by using a certain prediction algorithm. 
Let vector $X$ represent a video list generated by the server, where each element $x_i \in X$ at position $i$ represents video $x_i$ in $V$. For example, $x_5 =10$ means video 10 is to be played in the fifth order.
Here, we assume $|X| = |V|$. 
The server sends videos to the user in the order in $X$, one video at a time.
To calculate an optimized video list for the viewing, the user is assumed to watch each video $v$ according to its viewing time $\tau_v$, and when the time runs out, the user immediately switches to the next video. 
We assume that the server can immediately receive the video switch request and terminate the current video’s transmission and start the next video’s transmission, without any delay.

Then, we define some symbols for characterizing the transmission model in Fig.~\ref{FIG:transmissionmodel}. Let $\hat{r}$ and $r_v$ denote the burst rate for the initial segment and normal rate for the rest segments at the server side, respectively, where $r_v$ is also the video encoding rate of video $v \in V$.
Let $C$ and $\mu$ denote the token bucket capacity and the token rate, respectively. 
Let $K_{x_i}$ denote the number of available tokens in the token bucket when the video $x_i$ at position $i$ of $X$ starts to be transmitted.

In this paper, we are aimed to minimize the maximum startup delay of all videos (denoted by $D$), which is defined as follows. 
\begin{equation}
    \centering
    D=\max _{i=1, 2, \cdots, |X|}{d_{x_i}},
    \label{Equ.D}
\end{equation}
where $d_{x_i}$ is the startup delay of video $x_i$ and can be calculated by Eq.~(\ref{Equ.d_{x_i}}) given the initial available tokens $K_{x_i}$ when video $x_i$ starts to be transmitted at the server side.
\begin{subequations}
    \begin{numcases}
        {d_{x_i}=}
            B_{x_i}/\ \hat{r}, \quad {\ K}_{x_i}+\mu B_{x_i}/\ \hat{r}\ \geq B_{x_i}. \label{Equ.d_{x_i}-a} \\
            (B_{x_i}-{\ K}_{x_i})/\mu, \quad $otherwise.$ \label{Equ.d_{x_i}-b}
    \end{numcases}    
    \label{Equ.d_{x_i}}%
\end{subequations}

According to Eq.~(\ref{Equ.d_{x_i}-a}), the startup delay $d_{x_i}$ equals the time the server bursts out the initial segment $B_{x_i}$ if the number of the total accumulated available tokens (i.e., sum of the number of initial available tokens and number of newly generated tokens during the delivery of the initial segment) is larger than or equal to $B_{x_i}$, i.e., is sufficient for the token bucket to send the initial segment at the burst rate; Otherwise, $d_{x_i}$ is calculated by Eq.~(\ref{Equ.d_{x_i}-b}), which actually equals the time required for sending an amount of data $B_{x_i}-{\ K}_{x_i}$ at the token rate $\mu$. By definition, it is the sum of the time for sending certain part of the initial segment of $x_i$ at the burst rate until all the tokens (including the initially available tokens and also the tokens newly generated during this period of time) run out and the time for sending the remaining initial segment at the token rate.
For the correctness of derivation of Eq.~(\ref{Equ.d_{x_i}-b}), please refer to Appendix~\ref{appendix_a} for detailed proof.

Moreover, the initial available tokens ${\ K}_{x_i}$ used in Eq.~(\ref{Equ.d_{x_i}}) is recursively calculated as follows, given the initial available tokens ${\ K}_{x_{i-1}}$ of the previous video $x_{i-1}$.
\begin{equation}
    \begin{aligned}
        K_{x_i} = \min(& C,\\
        & K_{x_{i-1}}-(B_{x_{i-1}}-\mu d_{x_{i-1}})+
        (\mu\tau_{x_{i-1}}-\\
        &\min(\tau_{x_{i-1}}r_{x_{i-1}}, T_{x_{i-1}}r_{x_{i-1}}-B_{x_{i-1}}))).
    \end{aligned}      
    \label{Equ.K_{x_i}}
\end{equation}

The above equation indicates that $K_{x_i}$ equals the previous video's initial available tokens $K_{x_{i-1}}$ minus the net token consumption $B_{x_{i-1}}-\mu d_{x_{i-1}}$ during the previous video's burst transmission plus the net token increment $\mu\tau_{x_{i-1}}-\min(\tau_{x_{i-1}}r_{x_{i-1}}, T_{x_{i-1}}r_{x_{i-1}}-B_{x_{i-1}})$ during the previous video's playback, while $K_{x_i}$ cannot be larger than $C$. 

Finally, we formulate the optimized video ordering problem for finding a video list $X$ to minimize the maximum startup delay $D$ of all videos as follows.
\begin{equation} 
    \begin{aligned}
        &\min_{X}{D}  \\        
        s.t. \quad & |X|=|V|\\
        & x_i \in V, \forall i \in [1,|X|],
    \end{aligned}
    \label{Equ.formulatedproblem}
\end{equation}
where the first constraint requires that the size of the video list $X$ should be the same to the size of the video set $V$, and the second constraint requires that any video in the video list $X$ should be a video in $V$.

\subsection{Intractability Analysis}

In this subsection, we prove the optimized video ordering problem formulated in the preceding subsection is NP-hard by reducing it to the number-partitioning problem \cite{NumberPartitioning}.

Before proving this claim, we first introduce some concepts. 
We first introduce the concept of “minimal required initial tokens’’ required for a video $v$ to deliver its initial segment at full burst rate, denoted by $P_v$, which is the minimal amount of required tokens in the bucket when the initial segment of video $v$ starts to be transmitted such that $P_v+B_v \mu/\hat{r}$ can exactly support the full delivery of the initial segment of $v$ at the burst rate while the amount of tokens in the bucket drops to zero at the end of its delivery. 
The net token increment caused by a video $v$’s delivery (denoted by $\delta_v$) is defined as the amount of tokens accumulated at the token bucket during the video’s playback as the video playback rate is less than the token rate. 
Next, we introduce the concepts of positive- and negative-gain videos.
A video $v$ is called a positive-gain video, if the condition $\delta_v \ge P_v$ holds, where the net token increment $\delta_v$ is calculated as $\delta_v=\mu\tau_v-\min(\tau_vr_v, T_v r_v-B_v)$ according to Eq.~(\ref{Equ.K_{x_i}}) and the minimal required initial tokens $P_v$ is calculated as $P_v = B_v - \mu B_v/ \hat{r}$ according to Eq.~(\ref{Equ.d_{x_i}}). Otherwise, it is called a negative-gain video.

\begin{theorem}\label{Intractability}
The optimal video ordering problem formulated in Eq.~(\ref{Equ.formulatedproblem}) is NP-hard.
\end{theorem}

\begin{proof}

To prove this claim, we construct the following video set $V$, which will be used in this proof. The constructed video set $V$ consists of $M$ positive-gain videos and $M \cdot Y+Y+1$ negative-gain videos, $M>1, Y>2$. 
Moreover, $V$ satisfies the following three conditions: 
\begin{itemize}

\item \textbf{Condition I}. All videos have the same initial segment size $B$, bit rate $r$, and thus have the same minimal required initial tokens $P$ for smoothly transmitting their respective initial segments without causing extra startup delay (comparing with the case of delivery at full burst rate);   

\item \textbf{Condition II}. Delivery of any of the positive-gain videos can fully fill the token bucket, which is set with a fixed capacity $C = Y \times P$, such that a positive-gain video can accommodate all its subsequent $Y$ consecutive videos' initial segments’ burst transmissions without causing extra startup delay; 

\item \textbf{Condition III}. The total net token increment by any $Y$ (consecutive) negative-gain videos in $V$ is smaller than $P$ while larger than that by any individual negative-gain video. Let $\delta_{\text{Y-NV}}^{\textsc{min}}$ denote the minimal total net token increment by any $Y$ (consecutive) negative-gain videos and $\delta_{\text{1-NV}}^{\textsc{max}}$ denote the maximum net token increment by any individual negative-gain video. We have $\delta_{\text{1-NV}}^{\textsc{max}}<\delta_{\text{Y-NV}}^{\textsc{min}}<P$.  
\end{itemize}

So far, we have finished the construction of the video set $V$ and also provided the token bucket related parameter required in our proof. Besides, we assume the initial token bucket is full before delivering video set $V$. 

Next, we shall prove that the video list with the minimal maximum startup delay for the video set $V$ meet the following form: First, all those negative-gain videos are arranged to form an initial list. Then, after each $Y$ consecutive negative-gain videos, starting from the very first video in the initial list, a positive-gain video is inserted. This operation is repeated until no positive-gain video remains. We accordingly obtain a video list as shown in Fig.~\ref{optimalvideolistform}. Obviously, the number of negative-gain videos after the last positive-gain video in the above list is $Y+1$.

\begin{figure}[t]
  \centering
  \subfloat[Form that an optimal video list must satisfy.\label{optimalvideolistform}]{\includegraphics[width=0.45\textwidth]{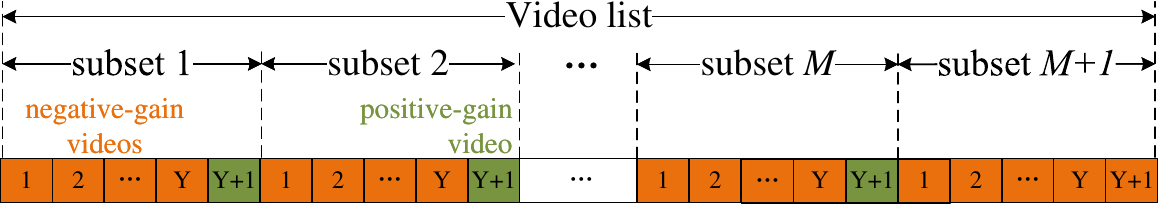}}\\
  \subfloat[Form that an optimal video list must not satisfy.\label{othervideolistform}]{\includegraphics[width=0.45\textwidth]{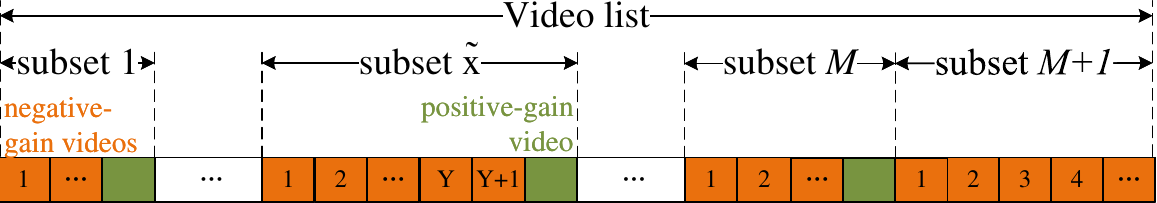}}\\
  \label{videolistform}
  \caption{Different video lists for the constructed video set $V$.}
\end{figure}

For any video list satisfying the form shown in Fig.~\ref{optimalvideolistform}, it can be easily derived that extra startup delay can (only) appear at each of the positive-gain videos and also the last negative-gain video in the list, whose extra startup delays are all smaller than or equal to $ (P-\delta_{\text{Y-NV}}^{\textsc{min}})/\mu$, and the exact maximum extra startup delay for this video list can be calculated as $(P - \min_{1 \le x \le M+1}\delta_{x[1:Y]})/\mu$, where $x \in [1, M+1]$ represents the subset number in the video list and $\delta_{x[1:Y]}$ represents the total net token increment by the first $Y$ negative-gain videos in subset $x$. 

For any video list that does not satisfy the form shown in Fig.~\ref{optimalvideolistform}, there must exist a subset (denoted by $\tilde{x}$), which contains no less than $Y+2$ videos and the $Y+2$-th video in this subset $\tilde{x}$ is preceded by $Y+1$ consecutive negative-gain videos (see Fig.~\ref{othervideolistform} for an example). In the subset $\tilde{x}$, both videos $Y+1$ and $Y+2$ will suffer extra startup delays and thus the extra startup delay of video $Y+2$ is determined by the net token increment only by video $Y+1$ and can be calculated as $(P - \delta_{\tilde{x}[Y+1]})/\mu$, where $\delta_{\tilde{x}[Y+1]}$ represents the net token increment by the negative-gain video $Y+1$ in subset $\tilde{x}$. 

According to Condition~III: $\delta_{\text{1-NV}}^{\textsc{max}}<\delta_{\text{Y-NV}}^{\textsc{min}}<P$, we have $\delta_{\tilde{x}[Y+1]}<\min_{1 \le x \le M+1}\delta_{x[1:Y]}<P$ and thus have $(P - \delta_{\tilde{x}[Y+1]})/\mu > (P - \min_{1 \le x \le M+1}\delta_{x[1:Y]})/\mu$. Therefore, a video list with the minimal maximum startup delay must meet the form shown in Fig.~\ref{optimalvideolistform}.


Finally, we prove that finding an optimal video list for the video set $V$ is NP-hard.
As analyzed above, an optimal video list must satisfy the form shown in Fig.~\ref{optimalvideolistform} and the maximum extra startup delay for a video list meeting such a form will be $(P - \min_{1 \le x \le M+1}\delta_{x[1:Y]})/\mu$. 
As a result, the problem to find an optimal video list is equivalent to finding an optimal arrangement of all negative-gain videos for a video list shown in Fig.~\ref{optimalvideolistform}, such that the video list has the maximized $\min_{1 \le x \le M+1}\delta_{x[1:Y]}$.
Obviously, this problem can be reduced to the number-partitioning problem, which is known to be NP-hard \cite{NumberPartitioning}. To the end, it is proved that solving the optimal video ordering formulated in Eq.~(\ref{Equ.formulatedproblem}) is NP-hard. 
\end{proof}

\section{Proposed PSAC Algorithm}
\label{psac}
In this section, we propose a Partially Shared Actor-Critic learning algorithm (PSAC) to generate an optimized video list for effectively reducing the maximum startup delay. We first give an overview of PSAC, and then present its algorithm design details.

\subsection{Overview}

PSAC works to find an optimized video list for effectively reducing the maximum startup delay through reinforcement learning. 
Fig.~\ref{FIG:psac} gives a flowchart regarding how PSAC is trained to generate a video list for effectively reducing the maximum startup delay. In terms of the structure of PSAC, it contains an actor and a critic, both of which share some common modules, including the embedding module and the encoder module, in order to reduce the amount of neural network parameters and therefore improve the convergence speed and performance.
In PSAC, the actor is to generate a video list for a given video set and calculate the maximal startup delay of the generated video list according to the token bucket based transmission model in the preceding section, while the critic is trained to output an estimated value to approximate the maximal startup delay achieved by the actor. 
The PSAC is to be trained by a real short video viewing trace, where the actual viewing time of each video is used as its predicted viewing time for the training, and the actor and the critic of PSAC are optimized by using policy gradient and stochastic gradient descent aiming to learn an optimal policy that can generate a video list striving to minimize the maximum startup delay. 

\begin{figure}[htbp]
	\centering		
    \includegraphics[width=0.45\textwidth]{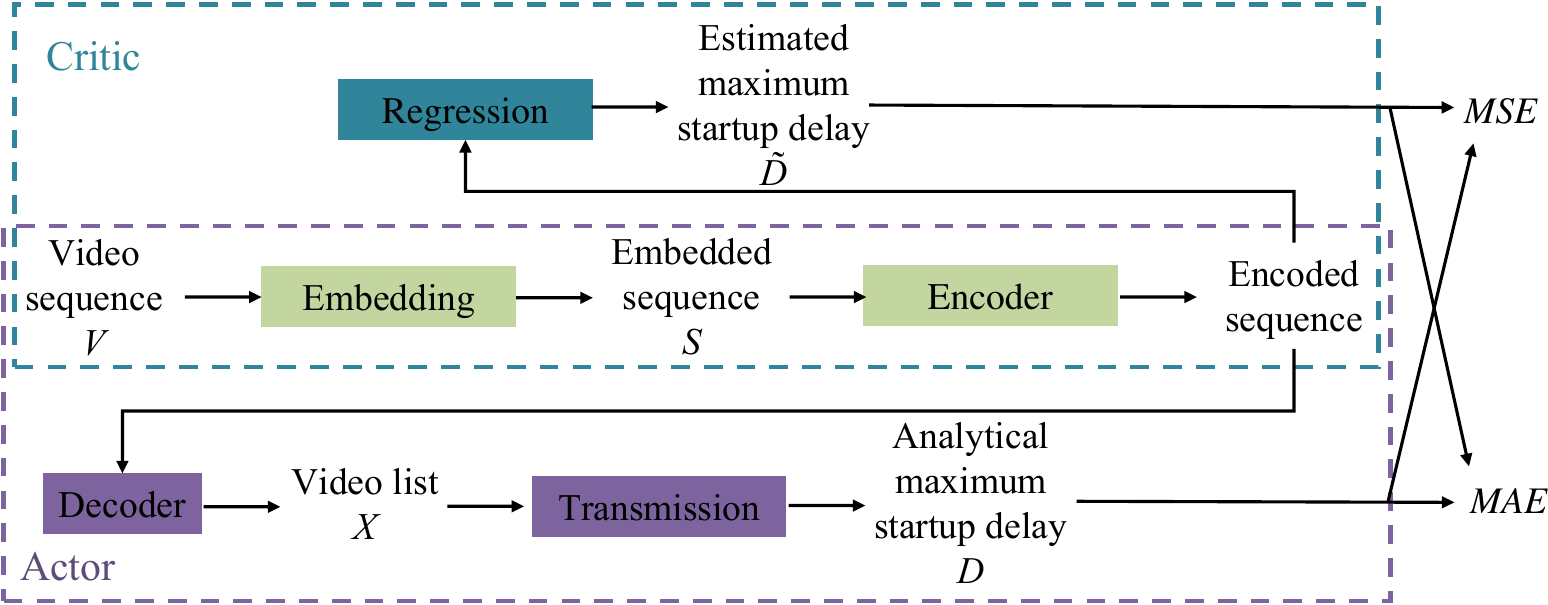}
	\caption{The structure of PSAC algorithm.}
	\label{FIG:psac}
\end{figure}

\subsection{Design}

As shown in Fig.~\ref{FIG:psac}, PSAC consists of actor and critic, where the actor consists of four modules, namely, embedding, encoder, decoder, and transmission modules, and the critic consists of three modules, where the embedding module and encoder module are shared with the actor, and the regression module is independent.

The two shared modules are designed as follows. 
The embedding module is a fully connected layer, and it transforms each video $v$ in the input video set $V$ to an $m$-dimensional vector, where each input video is a 4-dimensional vector including the video's duration, encoding rate, initial segment size, and viewing time, such that an input video set $V$ is embedded to an embedded sequence $S\in \mathbb{R}^{N\times m}$.
The encoder module is a recurrent neural network (RNN) consisting of $N$ Long Short-Term Memory (LSTM) cells each with a hidden size $m$, and it transforms the embedded sequence $S$, one video at a time step, into a sequence of latent memory states ${[{\textit{enc}}_1, {\textit{enc}}_2, \cdots, {\textit{enc}}_{N}]} \in \mathbb{R}^{N\times m}$,  which is called encoded sequence for short, to further characterize the latent relationship among videos.

In the actor, the decoder module is designed as an RNN with a pointing mechanism \cite{vinyals2015pointer} for converting the encoded sequence into a video list $X$.
The structure of the RNN is the same as that of the encoder and the pointing mechanism is implemented by the attention method \cite{bahdanau2014neural}. Then, the video list $X$ is input into the transmission module in the actor, which calculates and outputs the maximal startup delay $D$ of all videos in $X$ according to the token bucket transmission model depicted by Eq.~(\ref{Equ.D}). 

In the critic, the regression module is designed as a two-layer neural network with one $m$-neuron hidden layer and 1-neuron output layer. For the input encoded video sequence, the critic is trained to output an estimated value (denoted by $\tilde{D}$) to approximate $D$ output by the actor. 
To ease the description, we here call $D$ and $\tilde{D} $ as analytical maximum startup delay and estimated maximum startup delay, respectively.

Finally, based on the analytical and estimated maximum startup delays, $D$ and $\tilde{D}$, PSAC calculates their mean square error (MSE) and mean absolute error (MAE), and use them as the loss functions of the critic and the actor to train the modules of embedding, encoder, decoder, and regression.

\subsection{Training}
The training procedure of PSAC is given in Algorithm~\ref{alg:PSAC}.
Given the training set $\mathcal{V}$, the total number of training steps $E$, the batch size $Q$, it first initializes the actor parameter ${\theta}_a$ and the critic parameter ${\theta}_c$, and then iteratively trains the actor and critic by lines 2-13.
In each training step, it works as follows.
Line~3 randomly samples $Q$ video sets $V_1,V_2,\cdots,V_Q$ from the training set $\mathcal{V}$ as a training batch, each $V_i$ including $N$ videos. 
Following that, line~4 embeds each video set $V_i$ into an embedded sequence $S_i$, line~5 uses the encoder to transform each embedded sequence $S_i$ into an encoded sequence, and line~6 uses the decoder to decode each encoded sequence into a video list $X_i$.
For each $X_i$, lines~7 and~8 output the analytical maximum startup delay $D_i$ calculated according to Eq.~(\ref{Equ.D}) and the estimated maximum startup delay $\tilde{D}_i$ output by the regression module, respectively.
Lines~9 and~10 calculate the MAE $\mathcal{L}(\theta _a)$ and MSE $\mathcal{L}(\theta _c)$ based on all $D_i$ and $\tilde{D}_i$, respectively. 
Finally, lines~11 and~12 optimize the parameters of the actor and critic of PSAC, respectively, with the Adam optimizer~\cite{kingma2014adam}.
The above training process is carried out iteratively until the required training steps are reached.

\begin{algorithm}[htbp]
\caption{Training procedure in PSAC.}\label{alg:PSAC}
\begin{algorithmic}[1]
\renewcommand{\algorithmicrequire}{\textbf{Input:}}
\renewcommand{\algorithmicensure}{\textbf{Output:}}
\REQUIRE{Training set $\mathcal{V}$, training steps $E$, batch size $Q$;}
\ENSURE{$\theta _a$;}
\STATE Initializing model parameters ${\theta}_a$ and ${\theta}_c$ of the actor and the critic, respectively;
\FOR{$step$ =1, 2,$\cdots$, $E$}    
    \STATE $V_i \leftarrow $ SampleInput($\mathcal{V}$), for $i =1,2,\cdots Q$;
    \STATE $S_i \leftarrow $ Embedding($V_i$), for $i =1,2,\cdots Q$;
    \STATE ${[{enc}_1, {enc}_2, \cdots, {enc}_{N}]}_i \leftarrow$ Encoder($\theta _a, S_i$), for $i =1,2,\cdots, Q$;
    \STATE $X_i \leftarrow $  Decoder($\theta _a, {[{\textit{enc}}_1, {\textit{enc}}_2, \cdots, {\textit{enc}}_{N}]}_i$), for $i =1,2,\cdots Q$;  
    \STATE $D_i \leftarrow $ TransmissionModel($X_i$), for $i =1,2,\cdots, Q$; 
    \STATE $\tilde{D}_i \leftarrow $  Regression($\theta _c, {[{\textit{enc}}_1, {\textit{enc}}_2, \cdots, {\textit{enc}}_{N}]}_i$), for $i =1,2,\cdots, Q$;     
    \STATE $\mathcal{L}(\theta_a) \leftarrow \frac{1}{Q} \sum_{i=1}^{Q} (D_i- \tilde{D}_i) $;
    \STATE $\mathcal{L}(\theta_c) \leftarrow \frac{1}{Q} \sum_{i=1}^{Q} {(D_i- \tilde{D}_i)}^2$;
    \STATE $\theta _a \leftarrow \rm{ADAM}(\theta _a, \nabla _{\theta _a }\mathcal{L}(\theta _a))$;
    \STATE $\theta _c \leftarrow \rm{ADAM}(\theta _c, \nabla _{\theta _c }\mathcal{L}(\theta _c))$;   
\ENDFOR
\RETURN $\theta _a$.
\end{algorithmic}
\end{algorithm}
\section{Performance Evaluation}
\label{pe}
In this section, we evaluate the performance of our PSAC algorithm by comparing it with different algorithms by trace-driven simulations.

\subsection{Compared Algorithms}

Next, we give four algorithms for comparison with our PSAC algorithm.

\textit{1) Random ordering (RAND):} It randomly orders all the videos to generate a video list.

\textit{2) Interleaved ordering by viewing time (INTL).} It interleaves the videos according to their viewing time in a way such that the resulted video list has the following video viewing time order: smallest, largest, next smallest, next largest, and so on.

\textit{3) Greedy ordering (GRDY).} The idea behind this algorithm is to use a positive-gain video to restore the token bucket each time when the token bucket is about to run out of tokens. Specifically, given a video set, it first divides the positive-gain and negative-gain videos into a positive-gain video subset $X^+$ and a negative-gain video subset $X^-$, where the videos in either subset are both sorted in the ascending order of their net token increment, and then it iteratively inserts positive-gain videos in $X^+$ into $X^-$, one for each time, until all positive-gain videos are inserted into $X^-$. In each iteration, it selects a video with the smallest token net increase in $X ^+$ and inserts it into $X^-$ just before the last video with no extra startup delay (before which all videos have no extra startup delay). Finally, it outputs the $X^-$ as the generated video list.

\textit{4) None-module Shared Actor-Critic learning algorithm (NSAC).} This algorithm is very similar to our PSAC algorithm in structure except that the actor and critic of NSAC do not share any modules.

\subsection{Simulation Settings}

The simulation is based on a desensitized dataset provided by a well-known large-scale short video service provider in 2021, which contains more than 6 million viewing records of 16000 users, where each record contains multiple fields including user id, video id, video duration $T$, video encoding bitrate $\bar{r}$, and viewing time $\tau$. 
For our study, we selected 200 users from the dataset, and divided them into two groups of 160 users and 40 users and used their viewing records as the training set and test set, respectively. 
In addition, these 200 users were chosen as users with ranks in range $[16000 \times 10\% +1, 16000 \times 10\% +200] = [1601, 1800]$, which are ranked immediately after the top 10\% users having highest number of viewed videos (who may include anomalous users such as crawlers). Table \ref{tbl2} gives the key statistics of the 200-user dataset.

\begin{table}[width=\linewidth,cols=4,pos=h]
\caption{Key statistics for the 200-user dataset in our simulations.}\label{tbl2}
\begin{tabular*}{\tblwidth}{@{} LLLL@{} }
\toprule
 & Duration(s) & Viewing time(s) & Bitrate(Mbps) \\
\midrule
Mean & 29.32 & 24.3 & 1.86 \\
Std. deviation & 19.24 & 29.79 & 0.85\\
Minimum & 2 & 0.01 & 0.12\\
Quartile 1 & 13 & 3.32 & 1.23\\
Quartile 2 & 22 & 14.7 & 1.72\\
Quartile 3 & 48 & 33.12 & 2.44\\
Maximum & 62 & 299.75 & 5.59\\
\bottomrule
\end{tabular*}
\end{table}

For training the PSAC and NSAC, we set the training steps $E=20000$, the LSTM hidden size $m=128$, and the Adam optimizer's parameters $(\alpha, \beta_1, \beta_2, \epsilon)=(0.001, 0.9,$ $0.999, 10^{-7} )$ \cite{kingma2014adam}, and we set the batch size $Q=32$, which corresponds to 32 video sets, each consisting of $N$ videos randomly sampled from the viewed videos of a user in the training set, where $N$ is a simulation variable.
Moreover, for the transmission model of short video delivery, we set the burst rate $\hat{r} = 10$ Mbps, at which the server sends the initial segment of any video, while the normal rate is set to the bitrate of a requested video. We set the initial segment length of any video to be the same, i.e., 1 second long. 

For evaluating the performance of these algorithms, for each value of $N$, we prepared 256 video sets, each randomly sampled from the viewed videos of a user in the test set.
To implement an algorithm, we use it to rank each of the 256 video sets, compute the maximum video startup delays, and use the average of the 256 maximum video startup delays as the key performance of that algorithm.

We evaluate the performance of these algorithms under the following three scenarios: 

\begin{itemize}
\item All videos have accurate predicted viewing time (which equal their respective actual viewing time in the dataset) and also have the same bitrate of 2 Mbps. This scenario is used for highlighting the impact of different token bucket parameters (token bucket capacity $C$, token rate $\mu$, and number of the initial tokens $K_{x}$) and video set size $N$ on the startup delay performance by eliminating the impact of video bitrate variations and viewing time prediction error. 
\item All videos have accurate predicted viewing time while have different bitrates (which equal their respective actual bitrates in the dataset). This scenario is used for highlighting the impact of video bitrate variations. In this scenario, we used the actual video bitrates shown in Table~\ref{tbl2} to evaluate the performance of PSAC under varying bitrate.
\item All videos have inaccurate predicted viewing time (which equal their respective actual viewing time but plus a random variable following Gaussian distributions) while having a uniform bitrate of 2 Mbps. This scenario is used for highlighting the impact of viewing time prediction error. In this scenario, we evaluated the performance of these algorithms under different viewing time prediction errors, which are controlled by the standard deviation $\sigma$ of Gaussian distribution.
\end{itemize}

\subsection{Numerical Results}

\textit{1) Impact of token bucket capacity.} 
For this test, we changed the token bucket capacity $C$ from 2 to 10 Mbits with a step size of 1 Mbits, while fixing the token rate $\mu=2$ Mbps, the video set size $N=15$, and the amount of the initial tokens $K_{x} = C$. 
The numerical results are shown in Fig.~\ref{FIG:impactoftokenbucketcapacity}, where the $x$- and $y$-axes represent the ratio of token bucket capacity to the initial segment size, $C/B$, and the average maximum startup delay, respectively.
The curves marked by $\Diamond$, $\triangle$, $\bigcirc$, $\square$, and $\times$ are for RAND, INTL, GRDY, NSAC and PSAC, respectively.
It can be seen that the average maximum startup delay for all algorithms decreases with token bucket capacity increasing, while our PSAC always has the best performance. 
Specifically, PSAC can reduce the average maximum startup delay by up to 56\%, 54\%, 45\% and 37\% compared to RAND, INTL, GRDY and NSAC, which happens at $C/B$ = 2.5, 2.5, 2.5, 4, respectively.
The above results can be attributed to the following reasons. 
First, the use of multidimensional video features is more conducive to finding valuable video ordering clues, resulting in significantly higher performance for PSAC, compared to RAND and INTL, which do not use any video feature or only use one single video feature (viewing time).
Second, GRDY results in severe polarization in the startup delay performance, with the videos in the head part of a list having a very small startup latency and the videos in the rear of the list having a very large startup latency.
In contrast, PSAC tries to equalize the startup delays for different videos as much as possible, to maximally avoid extreme large startup delays.
Third, the partial module sharing design in PSAC results in obviously higher model inference performance than NSAC. 
Moreover, it can be seen the performance gain of PSAC varies with the token capacity, i.e., the gain is larger when $C/B \in$ [2,4] and becomes smaller beyond this region. The reason is as follows. 
First, too small the token capacity causes poor burst transmission capability, making it difficult to effectively cope with the subsequent burst transmission demands under any order of short videos. 
Second, excessively large token capacity can result in enough tokens accumulated during the watching of a long-viewing-time video, so as to satisfy the subsequent burst transmission demands under any order of short videos.
However, large token capacity may consume a lot of network resources and can be prohibitively costly for operators to deploy.

\begin{figure}[htbp]
	\centering		\includegraphics[width=0.45\textwidth]{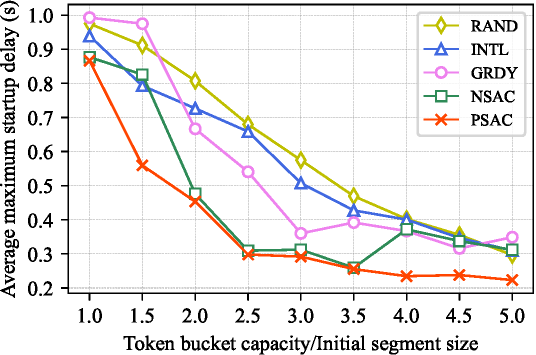}
	\caption{Impact of token bucket capacity $C$ when $\mu=2$ Mbps, $N$=15, and $K_{x}=C$.}	\label{FIG:impactoftokenbucketcapacity}
\end{figure}
\begin{figure}[htbp]
	\centering		
    \includegraphics[width=0.45\textwidth]{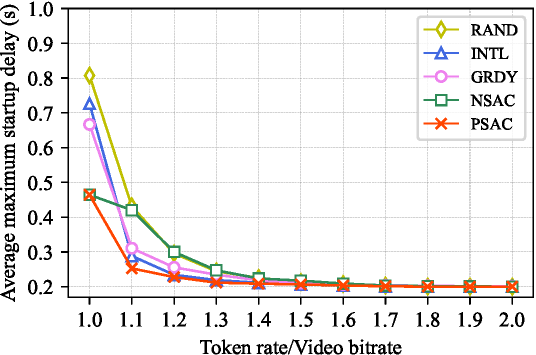}
	\caption{Impact of token rate $\mu$ when $C=4$ Mbits, $N=15$, and $K_{x}=4$ Mb.}
	\label{FIG:impactoftokenrate}
\end{figure}
\begin{figure}[htbp]
	\centering		\includegraphics[width=0.45\textwidth]{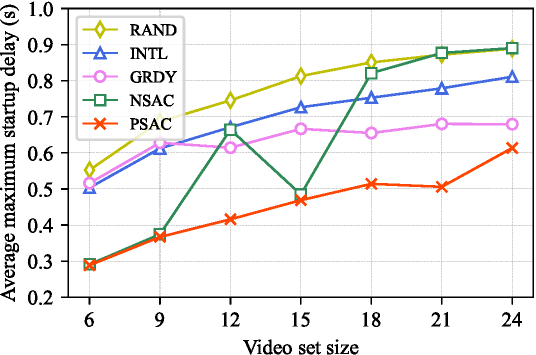}
	\caption{Impact of video set size $N$ when $C=4$ Mbits, $\mu=2$ Mbps, and $K_{x}=4$ Mb.}
	\label{FIG:impactofvideosetsize}
\end{figure}
\begin{figure}[htbp]
	\centering		\includegraphics[width=0.45\textwidth]{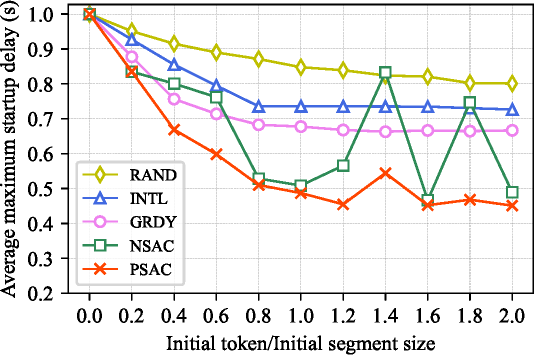}
	\caption{Impact of amount of initial tokens $K_{x}$ when $C=4$ Mbits, $\mu=2$ Mbps, and $N=15$.}
	\label{FIG:impactofinitialtoken}
\end{figure}
\begin{figure}[htbp]
	\centering		\includegraphics[width=0.45\textwidth]{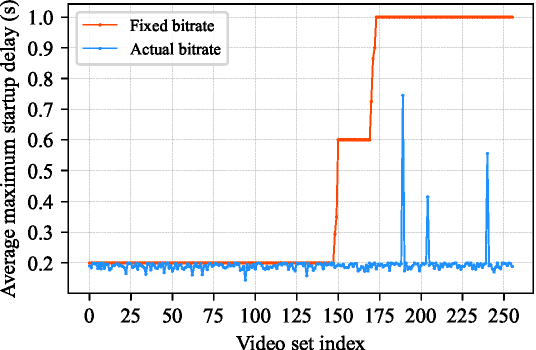}
	\caption{Performance of PSAC under actual and fixed bitrates.}
	\label{FIG:impactofbitrate}
\end{figure}
\begin{figure}[htbp]
	\centering		\includegraphics[width=0.45\textwidth]{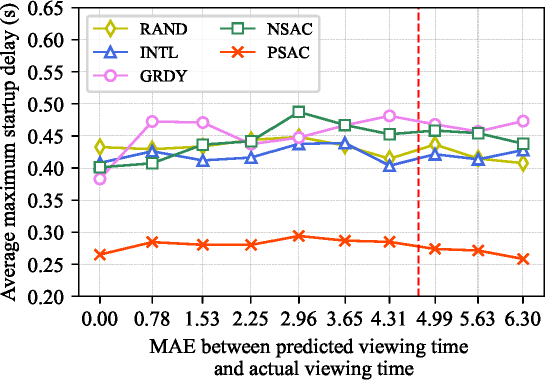}
	\caption{Impact of errors in viewing time prediction when $C=8$ Mbits, $K_x=8$ Mbits, $\mu=2$ Mbps, and $N=15$.}	\label{FIG:impactofviewingtimepredictionerror}
\end{figure}

\textit{2) Impact of token rate.} For this test, we changed the token rate $\mu$ from 2 to 4 Mbps with a step size of 0.2 Mbps, while fixing the token bucket capacity $C=4$ Mbits, video set size $N=15$, and the amount of the initial tokens $K_{x}=4$ Mbits. 
The numerical results are shown in Fig.~\ref{FIG:impactoftokenrate}, where the $x$- and $y$-axes represent the ratio of token rate to bitrate $\mu/\bar{r}$ and the average maximum startup delay, respectively.
Fig.~\ref{FIG:impactoftokenrate} shows that, as the token rate increases, the average maximum startup delays of all algorithms gradually converges to the same level of 0.2s, which is the time the server bursts out an initial segment.
Specifically, when the token rate is large enough, say $\mu/\bar{r} > 1.6$, the video startup delay is largely determined by the transmission model itself while almost irrelevant to how the videos are ordered.
However, when the token rate is low, say $\mu/\bar{r} \le 1.3$, PSAC achieves significantly lower startup delay. Specifically, PSAC can respectively reduce the average maximum startup delay by 43\%, 36\% and 30\% compared to RAND, INTL and GRDY at $\mu/\bar{r}=1$, and can reduce the average maximum startup delay by 40\% compared to NSAC at $\mu/\bar{r}=1.1$.
The insight of the above results is that, the smaller the token rate, the slower the rate of token accumulation, and thus the more the necessity to use a high-performance video ordering algorithm like PSAC to properly arrange the video order to accumulate more tokens for subsequent video burst transmissions so as to improve the startup performance.

\textit{3) Impact of video set size.}
For this test, we changed the video set size $N$ from 6 to 24 with a step size of 3, while fixing $C=4$ Mbits, $\mu=2$ Mbps, and $K_{x}=4$ Mbits. 
The numerical results are shown in Fig.~\ref{FIG:impactofvideosetsize}, where the $x$- and $y$-axes represent the video set size and the average maximum startup delay, respectively.
Specifically, PSAC can reduce the average maximum startup delay by up to 48\%, 42\%, 44\%, and 42\% compared to RAND, INTL, GRDY, and NSAC, which happens at $N=6, 6, 6, 21$, respectively. 
Such large performance gap demonstrates the significant advantage of our PSAC algorithm in dealing with complex situations because of the following reasons: As $N$ increases, there will be more space for optimizing/reducing the startup delay, and our algorithm has the ability to find a better video order from this increased space.

\textit{4) Impact of amount of initial tokens.}
For this test, we changed the amount of initial tokens $K_{x}$ from 0 to 4 Mbits with a step size of 0.4 Mbits, while fixing $C=4$ Mbits, $\mu=2$ Mbps, and $N=15$. 
The numerical results are shown in Fig.~\ref{FIG:impactofinitialtoken}, where the $x$- and $y$-axes represent the ratio of initial token to the initial segment size, $K_{x}/B$, and the average maximum startup delay, respectively.
Specifically, at $K_{x}/B=0$, the average maximum startup delay of all algorithms is the same, which is contributed to the fact that the first video in the video list will inevitably experience the maximum 1-second startup delay. At $K_{x}/B=1.2$, PSAC can reduce the average maximum startup delay by up to  46\%, 38\% and 32\% compared to RAND, INTL, and GRDY respectively. By comparing NSAC and PSAC, PSAC can reduce the average maximum startup delay by up to 34\% at $K_{x}/B=1.4$, which demonstrates that partial module sharing can make the model easier to converge, resulting in a more stable and superior performance.

\textit{5) Performance of PSAC under fixed and actual bitrates.} 
In this test, we evaluated the performance of PSAC under fixed and actual bitrates using the same 256 video sets.
For this test, $C=4$ Mbits, $\mu=2$ Mbps, $N=15$, and  $K_{x}=4$ Mbits.
Moreover, we removed the viewing records with video bitrates higher than the 2-Mbps token rate.
For the testing results, we sorted and numbered these video sets according to their maximum startup latency under the fixed bitrate.
The numerical results are shown in Fig.~\ref{FIG:impactofbitrate}, where the $x$- and $y$-axes represent the video set index and the maximum startup delay, respectively. 
It can be seen that the maximum startup delay under fixed bitrate (say orange curve) is an upper bound to the maximum startup delay under the actual bitrates (say blue curve), which is mainly attributed to the fact that the actual bitrates of all videos are always smaller than or equal to the fixed bitrate we set. 
On the one hand, for the transmission of the initial segment of a video, a smaller video bitrate means a smaller initial segment size and a shorter time for the server to send out the initial segment. 
On the other hand, for the transmission of the remaining video segments, a smaller video bitrate means smaller token consumption and a faster speed for the token bucket to accumulate tokens to cope with the burst transmission of subsequent videos.

\textit{6) Impact of errors in predicted viewing time.}
For this test, we used ten zero-mean Gaussian distributions, each with a different standard deviation $\sigma=0,1,\cdots,9$, to noise the actual viewing time of the videos to generate their respective predicted viewing time, in order to evaluate the robustness of PSAC to the viewing time prediction errors. All predicted viewing time is forced to be greater than 0.
In this test, we fixed $C=8$ Mbits, $K_x=8$ Mbits, $\mu=2$ Mbps, and $N=15$. 
The numerical results are shown in Fig.~\ref{FIG:impactofviewingtimepredictionerror}, where the $x$- and $y$-axes represent the MAE between the predicted viewing time and the actual viewing time, and the average maximum startup delay, respectively. The MAE value for each point in Fig.~\ref{FIG:impactofviewingtimepredictionerror} corresponds to a specific Gaussian distribution.
The red dash line indicates the MAE (4.741) that can be achieved by the existing prediction algorithm reported in the literature \cite{Lin2023tree}.
Fig.~\ref{FIG:impactofviewingtimepredictionerror} shows that PSAC always significantly outperform the other compared algorithms for any value of MAE. 
It further confirms that, compared to other algorithms, our PSAC has a stronger ability to mine valuable ordering clues from video features, and has a sufficiently high tolerance for different MAE values.
In particular, as shown in the figure, at MAE=4.99, which is near the MAE level that existing viewing time prediction algorithms can achieve, PSAC can reduce the average maximum startup delay by 37\%, 35\%, 40\%, and 41\%, respectively, compared to RAND, INTL, GRDY and NSAC.

\section{Conclusion}
\label{con}
In this paper, we studied short video ordering for reducing the startup delay. 
We first modeled the short video transmission path as a token bucket and proposed the idea to interleave short-viewing-time videos and long-viewing-time videos in a video list to make the tokens consumed by each transmission of a short-viewing-time video replenished by the next long-viewing-time video’s transmission as much as possible so as to allow each video to have a small startup delay. 
We formulated the video ordering problem as a combinatorial optimization problem, proved its NP-hardness, and proposed the PSAC reinforcement learning algorithm to learn an optimized video ordering strategy to transform an input video set into a video list for achieving minimized video startup delay. 
Numerical results based on a real dataset demonstrated that the proposed PSAC algorithm can significantly reduce the startup delay of videos compared to baseline algorithms.

In the future, we shall explore how to integrate the idea of video ordering and prefetching together for further improved startup delay performance while minimizing the potential downloaded data waste caused by randomized video switching at user side. 

\appendix
\section{Proof of Eq.~(\ref{Equ.d_{x_i}-b})}
\label{appendix_a}

Here we give the proof of Eq.~(\ref{Equ.d_{x_i}-b}), i.e., the video startup delay $d_{x_i} =(B_{x_i}-K_{x_i})/\mu$
in the case  $K_{x_i}+\mu B_{x_i}/\ \hat{r}< B_{x_i}$.
\begin{proof}
In this case, the token bucket decomposes the transmission of the initial segment into two stages with different rates: The burst rate and the token rate, respectively. Obviously, the moment the tokens run out is the time instant separating these two stages.

Without loss of generality, assume that the initial segment’s transmission starts at time instant 0. Suppose the time instant when the token bucket is exhausted is $t$, and it should satisfy $\hat{r}t = K_{x_i}+\mu t$, where the left side is the data amount sent out by the server by time $t$, and it just exhausts the totally   accumulated available tokens that are shown in the right side. Therefore, we have 
\begin{equation}
    t= K_{x_i}/(\hat{r}-\mu).
    \label{Equ.appendix_t}
\end{equation}

After time $t$, the token bucket transmits the remaining data of the initial segment, $B_{x_i}-\hat{r}t$, at the token rate $\mu$, and the transmission time $t'$ is calculated as follows.
\begin{equation}
    t' = (B_{x_i}-\hat{r}t)/\mu.
    \label{Equ.appendix_t'}
\end{equation}

Finally, the startup delay for video $x_i$ is calculated as follows according to Eqs.~(\ref{Equ.appendix_t}) and (\ref{Equ.appendix_t'}).
\begin{equation}
    \begin{aligned}
        d_{x_i} &=t+t'\\
               &=(B_{x_i}-K_{x_i})/\mu
    \end{aligned}    
\end{equation}

Therefore, Eq.~(\ref{Equ.d_{x_i}-b}) is proved. 

\end{proof}
\printcredits

\bibliographystyle{elsarticle-num}

\bibliography{manuscript}

\end{document}